\documentclass[conference,a4paper]{APSIPA2021}
\usepackage{amsmath}
\usepackage{graphicx}
\usepackage{multirow}
\usepackage{threeparttable}
\usepackage{cite}
\usepackage{amsmath,amssymb,amsfonts}
\usepackage{graphicx}
\usepackage{textcomp}
\usepackage{xcolor}
\usepackage{mathrsfs}
\usepackage[noend]{algpseudocode}
\usepackage[ruled]{algorithm2e}
\usepackage{bbm}
\usepackage{subfigure}
\usepackage{url}  
\usepackage{tabularx}
\usepackage{booktabs}
\usepackage{comment}
\usepackage{flushend,cuted}
\usepackage{booktabs}
\usepackage{multirow}
\usepackage{pifont}
\newtheorem{theorem}{Theorem}  
\newtheorem{lemma}{Lemma}  
\newtheorem{proof}{Proof}	
\newtheorem{definition}{Definition}

\usepackage{fancyhdr}
\fancypagestyle{noheader}{
  \fancyhf{} 
  \fancyfoot[C]{\thepage} 
}

\usepackage{geometry}
\usepackage{fancyhdr}

\fancypagestyle{firststyle}{
  \fancyhf{}
  \fancyhead[C]{2023 Asia Pacific Signal and Information Processing Association Annual Summit and Conference (APSIPA ASC)}
}

\begin{document}


\title{Modeling and Analysis of the Epidemic-Behavior Co-evolution Dynamics with User Irrationality
}

\author{%
\authorblockN{%
Wenxiang Dong and 
H. Vicky Zhao
}
\authorblockA{%
Department of Automation, BNRist, Tsinghua University, Beijing, P. R. China 100084  \\
E-mail: dongwx19@mails.tsinghua.edu.cn; vzhao@tsinghua.edu.cn}
}

\maketitle
\thispagestyle{plain}
\pagestyle{fancy}
\begin{abstract}
During a public health crisis like COVID-19, individuals' adoption of protective behaviors, such as self-isolation and wearing masks, can significantly impact the spread of the disease. In the meanwhile, the spread of the disease can also influence individuals' behavioral choices. Moreover, when facing uncertain losses, individuals' decisions tend to be irrational. Therefore, it is critical to study individuals' irrational behavior choices in the context of a pandemic. In this paper, we propose an epidemic-behavior co-evolution model that captures the dynamic interplay between individual decision-making and disease spread. To account for irrational decision-making, we incorporate the Prospect Theory in our individual behavior modeling. We conduct a theoretical analysis of the model, examining the steady states that emerge from the co-evolutionary process. We use simulations to validate our theoretical findings and gain further insights. This investigation aims to enhance our understanding of the complex dynamics between individual behavior and disease spread during a pandemic.\footnote{The paper is accepted by 2023 Asia Pacific Signal and Information Processing Association Annual Summit and Conference (APSIPA ASC 2023)}
\end{abstract}
\begin{IEEEkeywords}
Irrationality, Epidemic, Individual behavior
\end{IEEEkeywords}
\section{Introduction}

The COVID-19 pandemic has triggered a profound public health crisis and inflicted substantial economic damage. Throughout such public health crisis, individuals' behavior and the pandemic influence each other in a complex way. For instance, individuals' behaviors, such as wearing masks, social distancing, and self-quarantine, can help limit the spread of the disease, while people tend to adopt protective behaviors more readily when the severity of the pandemic becomes apparent. Furthermore, individuals often display irrational behavior during the pandemic, including heightened panic at the early stage of an outbreak or a tendency to underestimate the severity of the disease as it spreads further. These irrational responses can significantly impact people's decision-making processes and ultimately influence the pandemic. Therefore, it is crucial to model individuals' irrational behavior and to analyze its impact on the pandemic.

\subsection{Literature Review}
There have been extensive works on user behavior modeling during an epidemic, which we will briefly introduce in the following. 

\subsubsection{Behavior Modeling in Disease Spread}
Several previous studies focus on modeling individual behavior during disease outbreaks. The authors of \cite{zhang2014suppression, wu2012impact, bagnoli2007risk} believe that as the number of infected individuals in a given environment increases, people are more inclined to adopt precautionary measures. Zhang et al. propose that individuals are more inclined to engage in protective behaviors when they have a higher number of infected neighbors \cite{zhang2014suppression}. Wu et al. posit that an individual's adoption of protective behaviors is influenced not only by the proportion of infected neighbors but also by the regional and global infection rates \cite{wu2012impact}. The works in \cite{funka2009spread, granell2013dynamical, granell2014competing, wang2019impact,zheng2018interplay} also propose that information dissemination plays a crucial role in shaping people's protective behavior during disease outbreaks, and individuals who receive relevant information are more likely to adopt precautionary measures in response to the pandemic. Funk et al. believe that infected individuals play a significant role in disseminating information about the disease across networks. This information dissemination can effectively encourage others to adopt protective measures in response to the ongoing spread of the disease \cite{funka2009spread}. Granell et al. develop a dual network model that incorporates both disease spread and information spread. In their model, one network represents the transmission of the disease, while the other network represents the dissemination of information. Their findings suggest that individuals who receive informational messages are more vigilant and, as a result, exhibit a lower infection rate \cite{granell2013dynamical}. However, these studies often overlook a common yet crucial factor: the influence of individual irrationality on decision-making during a pandemic. Understanding and accounting for the impact of irrational behavior are essential as they can greatly shape individuals' responses and actions during an ongoing outbreak.

\subsubsection{Irrational Behavior Modeling}
Individuals often make irrational decisions when faced with risks, such as the risk of being infected during a disease outbreak. The Prospect Theory \cite{kahneman2013prospect, kahneman2013choices}, an influential theory put forth by Nobel laureates Kahneman and Tversky, sheds light on this phenomenon. It says that individuals tend to overestimate the likelihood of low-probability events and underestimate the likelihood of high-probability events. Moreover, individuals demonstrate a heightened sensitivity to potential losses compared to gains \cite{tversky1992advances}. To effectively guide people's behavior during an epidemic and control the spread of the disease, it is of great importance to model such irrational behavior and analyze its impact on the pandemic. 
The works in \cite{oraby2015bounded, hota2019game, li2020perception} examine the influence of people's irrationality on their vaccination decisions.
Vaccination presents a fundamental behavioral choice problem where individuals must make a single binary decision on whether or not to receive a vaccine during an epidemic, wherein opting for vaccination eliminates the risk of infection. However, during an epidemic, individuals also have the option to engage in various protective behaviors, such as wearing masks, practicing hand hygiene, and self-isolation at home. These behaviors are discretionary and may change from time to time depending on the severity of the pandemic. Note that during an epidemic, individuals only decide once on whether to receive the vaccination. But for protective behaviors such as wearing masks, practicing hand hygiene, and self-isolation at home, they have to constantly decide whether to adopt them until the pandemic ends. The dynamic interplay between disease spread and individual behavioral choices occurs continuously throughout the outbreak, highlighting the complex relationship between the two. To our knowledge, no prior work modeled and analyzed the complex interplay between disease spread and behavioral choice with consideration of user irrationality.

\subsection{Our Contribution}
The main contributions of this paper are:

\begin{itemize}
    \item We develop an epidemic-behavior co-evolution model that provides a framework for analyzing the co-evolution of the pandemic and individuals' behavior. This model allows us to study the dynamic interplay between disease spread dynamics and human behavior.

    \item In our model, we incorporate the concept of individuals' irrationality and specifically model irrational behavior using the Prospect Theory. By considering irrational decision-making processes, we enhance the realism of the model and capture the complex human decision-making process during a pandemic.

    \item We conduct theoretical analyses of the co-evolution of the pandemic and individuals' behavior. By exploring the dynamics of the model, we identify and characterize the steady states, which provide insights into the long-term behavior patterns and their interplay with the pandemic.
\end{itemize}

\section{The epidemic-behavior co-evolution model}
In this section, we model the co-evolution of individual behavioral choices and disease spread during the pandemic. 
In our model, extending the model in \cite{poletti2009spontaneous}, we assume that individuals can choose between 2 possible behaviors: risky behavior $a_1$ (like going out) and conservative behavior $a_2$ (like self-isolation). Individuals would choose the behavior according to the severity of the epidemic, and these behavior choices will, in turn, affect the epidemic. 
We use an undirected graph to model the interpersonal connections among individuals. This graph captures the spread of the disease through the network, and in addition, neighbors in this network influence each other's behavior. To simplify the analysis, in this work, we assume a regular network consisting of $N$ nodes, with degree $\bar{k}$.

Our model contains two parts, the disease spread model and the behavior change model. Given the current actions of all users, the disease spread model quantifies how the pandemic will spread in the network; while given the current number of infected people and their neighbors' current actions, the behavior change model describes how users choose their protective behaviors. We will introduce them one by one in the following.

\subsection{The Disease Spread Model}

We adopt the classical Susceptible-Infected-Susceptible (SIS) model to characterize the spread of the disease. In this model, each individual can be in one of two health states: susceptible or infected. In each time slot, a susceptible individual has the potential to become infected by an infected individual, with a given infection rate. Conversely, infected individuals have a chance to recover at a certain recovery rate. To mitigate the risk of infection, susceptible individuals can choose to adopt various protective measures, such as home quarantine, wearing masks, or taking no action. In this study, our focus is primarily on the protective behavior of susceptible individuals. We make the assumption that individuals who are already infected will not engage in such protective behaviors. This assumption is based on the understanding that infected individuals may not have the same motivation or need to adopt protective measures since they are already infected.

Assume that the susceptible individuals have two possible behavioral options $a_1$ and $a_2$. Furthermore, we assume that all individuals who choose $a_1$ have the same infection rate, denoted as $\beta_1$, and similarly, all individuals who choose $a_2$ have the same infection rate $\beta_2$.
The recovery rate is assumed to be the same for all infected people and is denoted as $\gamma$. At time slot $t$, let $s(t)$ and $i(t)$ represent the fractions of susceptible and infected individuals, respectively.  Additionally, let $x_1(t)$ be the fraction of susceptible individuals adopting behavior $a_1$, and $x_2(t)$ represent the fraction adopting behavior $a_2$. Consequently, the dynamics of the disease spread can be described by the following differential equations:
\begin{equation}
\begin{split}
\dfrac{ds(t)}{dt} = &\gamma i(t)-s(t)i(t)\overline{\beta} \bar{k},
\\ \text{and} \quad \dfrac{di(t)}{dt} = &s(t)i(t)\overline{\beta} \bar{k}-\gamma i(t),
\end{split} \label{eq:SISmodeo}
\end{equation}
where $\overline{\beta}=\beta_1 x_1(t)+\beta_2 x_2(t)$. Note that $s(t)+i(t)=1$, (\ref{eq:SISmodeo}) can be rewritten as
\begin{equation}
\dfrac{di(t)}{dt} = i(t)(1-i(t))\overline{\beta} \bar{k}-\gamma i(t).
\label{eq:disease}
\end{equation}

\subsection{The Behavior Change Model}


The behavior change model quantifies the change in $\{ x_1(t),x_2(t) \}$, the proportion of susceptible individuals taking actions $\{ a_1,a_2 \}$ at time $t$. The change in $\{ x_1(t),x_2(t) \}$ may be due to two possible reasons. First, due to the change in the proportion of the infected individuals, or influenced by their neighbors, susceptible individuals select different actions at time $t$. We use $E_i(t)$ to denote the change in $x_i(t)$ due to such users' active behavior change, where $i=1,2$. Additionally, users' health states may change at time $t$, causing the change in $x_i(t)$. For example, a susceptible user taking action $a_i$ at time $t-1$ becomes infected at time $t$, or vice versa. We use $B_i(t)$ to denote the change in $x_i(t)$ due to the change in users' health states. Thus, we have
\begin{equation}
\dfrac{dx_i}{dt} = E_i(t)+B_i(t),\ i=1,2.
\label{eq:dx}
\end{equation}

In the following, we will discuss $E_i(t)$ and $B_i(t)$ one by one.
\subsubsection{Analysis of $E_i(t)$}

To model users' active behavior change due to their neighbors' influence and the percentage of  individuals currently being infected, 
we use evolutionary game theory to model individuals' decision-making processes. Evolutionary game theory offers a valuable tool for understanding how individuals make choices in response to environmental influences. The fundamental components of the evolutionary game theory include strategy, payoff, and strategy update rules, each of which will be introduced in detail in the following.

\textbf{Strategy:}
Each behavior $a_i$ is a strategy. In each time slot, $m$ percent of susceptible individuals are randomly selected as focal individuals. They observe their neighbors' behaviors and choose one possible strategy. Other susceptible individuals keep their actions unchanged.


\textbf{Payoff:}
Each individual would get a payoff based on their chosen action and their interactions with their neighbors. To model users' irrationality under conditions of uncertainty, we define the payoffs of different behaviors using the Prospect Theory (PT). 
The Prospect Theory seeks to explain how individuals make decisions under conditions of uncertainty. PT suggests that people do not always make rational choices, and their decisions can be influenced by various cognitive biases and subjective evaluations.
According to PT, if an individual selects a behavior $a_i$ that leads to $L$ potential outcomes $o_1, o_2, ..., o_L$ with corresponding probabilities $p_1, p_2, ..., p_L$, respectively, then the payoff $U_i$ associated with choosing $a_i$ can be defined as follows:
\begin{equation}
    U_i=\sum_{i=1...L}u(o_i)\pmb{\omega}(p_i),
    \label{PTequation}
\end{equation}
where $u(x)$ is the value function and $\pmb{\omega}(p)$ is the weighting function, which we will introduce one by one.

For the value function $u(x)$, according to the Prospect Theory, the perceived payoff differs from the actual payoff due to the presence of cognitive biases and subjective evaluations that individuals make when faced with uncertain outcomes. These biases and evaluations influence how individuals perceive and interpret information, leading to deviations from rational decision-making. The value function $u(x)$ models the relationship between the actual payoff $x$ and the perceived payoff $u(x)$. A widely used value function is the power function 
\begin{equation}
u(x)=
\left\{
\begin{aligned}
&x^\sigma,&& if\  x\ge 0\\
&-\lambda(-x)^\sigma,&& if\  x< 0\
\end{aligned}
\right.
\label{utilityfunc}
\end{equation}
where $\lambda$ reflects the individual's different sensitivity to gain and loss. The sensitivity coefficient $\sigma \in (0,1]$ determines the shape and curvature of the value function \cite{prelec2000compound}.


For the weighting function $\pmb{\omega}(p)$, when facing uncertainty, individuals' perceived probability is different from real probability. Individuals have a tendency to overestimate the probability of small risks, while underestimating the probability of larger risks \cite{1979Prospect}. According to the work in \cite{prelec1998probability}, the weighting function describing the relationship between the perceived probability and the real probability is:
\begin{equation}
\pmb{\omega}(p)=e^{(-(-lnp)^\alpha)},\ \ p\in[0,1],\ \alpha \in (0,1],
\end{equation}
where $p$ is the real probability of a certain outcome, $\pmb{\omega}(p)$ is the perceived probability, and $\alpha$ is the rationality coefficient. A smaller $\alpha$ indicates a higher degree of individuals' irrationality.


Fig.\ref{WeightFunction} plots the weighting function $\pmb{\omega}(p)$ with different $\alpha$. When the actual probability $p$ is very small, the perceived probability $\pmb{\omega}(p)$ is much larger than the actual probability $p$, showing the phenomenon that individuals tend to overestimate events/risks with small probabilities. On the contrary, when the actual probability is large, the perceived probability $\pmb{\omega}(p)$ is smaller than the actual probability $p$, modeling the phenomenon that individuals tend to underestimate events/risks with large probabilities. Also, from Fig.\ref{WeightFunction}, a smaller $\alpha$ (that is, a higher degree of individuals' irrationality) intensifies individuals' tendency to overestimate the probability of small risks while simultaneously underestimating the probability of larger risks.


    \begin{figure}[t]
	\centering
	\begin{minipage}[t]{0.8\linewidth}
		\centering
		\includegraphics[width=1\linewidth]{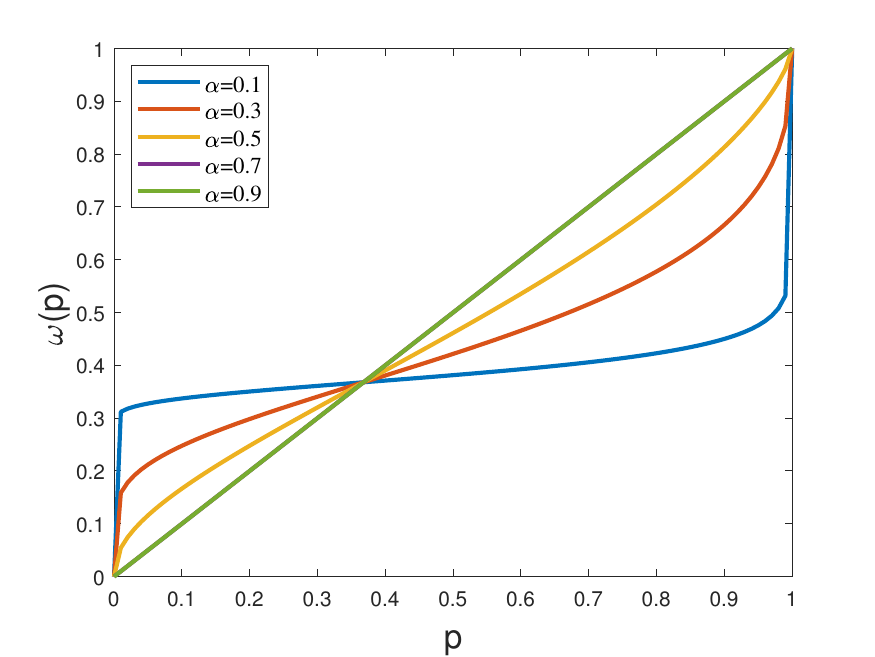}	
	\end{minipage}%
	\caption{The weighting function of different rational coefficients $\alpha$.}
	\label{WeightFunction}
\end{figure}


In our model, each susceptible individual who adopts behavior $a_i$ receives a payoff of $c_i$ with probability 1, which is determined by the behavior itself. Also note that an individual taking behavior $a_i$ has an approximate probability of $\bar{k}\beta_i i(t)$ to be infected at time $t$ \cite{0DIRECTED}, which incurs an additional loss of $c_n$ (we consider those epidemics with $\beta_i^t \ll 1$, such as SARS, MERS and common influenza\cite{wu2020nowcasting}, and we assume that $\bar{k}\beta_i < 1$).
To simplify the analysis, we assume that all individuals have the same form of the utility function $u(x)$ and the weighting function $\pmb{\omega}(p)$. Thus, from \eqref{PTequation}, the perceived payoff when choosing behavior $a_i$ is

\begin{equation}
U_i=u(c_n)\cdot\pmb{\omega}[\bar{k}\beta_i  i(t)]+u(c_i),\ i=1,2.
\end{equation}

\textbf{Strategy Update Rule:}
At each time step, $m$ percent of the susceptible individuals are randomly selected as focal individuals and update their actions, while the rest individuals' actions remain unchanged. These focal individuals are more likely to imitate the behaviors of individuals who have higher payoffs. Following the work in \cite{2004Coevolutionary}, a focal individual with strategy $a_i$ compares its payoff $U_i$ to the payoff $U_j$ of another randomly chosen individual with strategy $a_j$. The focal individual then updates its strategy to $a_j$ with probability:
\begin{equation} \label{eqn:actionchangeprob}
p(a_i\rightarrow a_j)=\frac{1}{2}+\frac{\omega}{2}\frac{1}{U_{max}}(U_j-U_i),
\end{equation}
where $\omega \in (0,1]$ measures the strength of selection, and $U_{max}$ is the normalization term.

For the two behavior $a_1$ and $a_2$, the probability that $x_1(t)$ increases by $\frac{1}{N_0(t)}$ is:
\begin{equation} \label{eqn:x1increaseby1}
p\left(\Delta x_1=\frac{1}{N_0(t)} \right) =x_1(t)x_2(t)p(a_2\rightarrow a_1),
\end{equation}
{where $N_0(t)$ is the number of susceptible individuals at $t$ and $N_0(t)=Ns(t)$.} Similarly, the probability that $x_1(t)$ decreases by $\frac{1}{N_0(t)}$ is:
\begin{equation} \label{eqn:x1decreaseby1}
p\left(\Delta x_1=-\frac{1}{N_0(t)}\right)=x_1(t)x_2(t)p(a_1\rightarrow a_2).
\end{equation}

Combining (\ref{eqn:actionchangeprob}), (\ref{eqn:x1increaseby1}) and (\ref{eqn:x1decreaseby1}), we have:
\begin{equation}
\begin{split}
E_1(t)=&mN_0(t)\left\{ p \left(\Delta x_1=\frac{1}{N_0(t)} \right) \times \frac{1}{N_0(t)} \right .\\
&\left .- p\left(\Delta x_1=-\frac{1}{N_0(t)} \right) \times \frac{1}{N_0(t)}\right \}\\
=&\frac{m\omega}{U_{max}}x_1(t)x_2(t) \left\{u(c_n)\cdot\pmb{\omega}[\bar{k}\beta_1  i(t)] \right.\\ 
&\left.+u(c_1)-u(c_n)\cdot\pmb{\omega}[\bar{k}\beta_2  i(t)]-u(c_2) \right \},
\label{eq:Ei}
\end{split}
\end{equation}
where $m$ is the fraction of individuals who are chosen as the focal individuals at each time.
We can get $E_2(t)$ in the same way, which is omitted here.

\begin{figure*}[t]
\begin{align}
&\dfrac{di}{dt} = i(t)(1-i(t))\overline{\beta} \bar{k}-\gamma i(t),\label{eq:M-model}\\
&\dfrac{dx_1}{dt} =x_1(t)x_2(t)\bar{k} i(t)(\beta_2-\beta_1) +\frac{m\omega}{U_{max}} x_1(t)x_2(t)(u(c_n)\cdot\pmb{\omega}[\bar{k}\beta_1  i(t)]+u(c_1)-u(c_n)\cdot\pmb{\omega}[\bar{k}\beta_2  i(t)]-u(c_2)).\nonumber
\end{align}
\underline{\hbox to \textwidth{}}
\end{figure*}

\subsubsection{Analysis of $B_i(t)$}

Even if individuals do not change their behavior, the proportion of different behaviors among susceptible individuals will change over time due to the change of their health states. 
Let $s_1(t)$ be the fraction of individuals who are both susceptible and adopt behavior $a_1$ at time $t$ among the entire population, that is, $s_1(t)=s(t)x_1(t)$ where $s(t)$ is the fraction of susceptible individuals and $x_1(t)$ is the fraction of individuals adopting $a_1$ among all the susceptible individuals. As $s(t)=s_1(t)+s_2(t)$, we have:
\begin{equation}
\begin{split}
B_{1}(t)=&\frac{d}{dt} \left ( \frac{s_{1}(t)}{s(t)}\right )
=\frac{d}{dt} \left ( \frac{s_{1}(t)}{s_1(t)+s_2(t)}\right )\\
=&\frac{s_{1}'(t)s_2(t)-s_{2}'(t)s_1(t)}{s^2(t)}.
\end{split}
\label{eq:Nit}
\end{equation}

Note that $s_1'(t)$, the first order derivative of $s_1(t)$, contains two parts. The first part is the change caused by the infection of susceptible users, while the second part is the change caused by the recovery of infected individuals. For the first part, we have $s_{1a}'(t) = -s_1(t)\beta_1 \bar{k} i(t)$, which represents the decrease in the fraction of individuals adopting behavior $a_1$ due to infection. For the second part, we assume that the recovered individuals would choose their behaviors based on the current ratio of different behaviors in the environment, as defined in \cite{2020Decisions}. Therefore, we have $s_{1b}'(t) = \gamma x_1(t)i(t)$, which represents the increase in the fraction of individuals adopting behavior $a_1$ due to the recovery of infected individuals. Combining these two parts, we obtain the overall change in the fraction of individuals adopting behavior $a_1$ as $s_1'(t) = s_{1a}'(t) + s_{1b}'(t)$. Then we have:
	\begin{equation}
	\begin{split}
	s_1(t)&=s(t)x_1(t),
	\\s_1'(t)
	&=-s(t)x_1(t)\beta_1 \bar{k} i(t)+\gamma x_1(t)i(t).
	\end{split}
	\label{eq:st}
	\end{equation}
 
We can get $s_2(t)$ and $s_2'(t)$ in the same way. Then based on \eqref{eq:Nit} and \eqref{eq:st}, we  have
	\begin{equation}
	B_1(t)=x_1(t)x_2(t)\bar{k} i(t)(\beta_2-\beta_1).
	\label{eq:Bi}
	\end{equation}
We can get $B_2(t)$ in the same way, which is omitted here.

Combining \eqref{eq:dx}, \eqref{eq:Bi} and \eqref{eq:Ei}, the complete differential equation describing the dynamics of individual behavior change is:
\begin{align}
    &\dfrac{dx_1}{dt} =x_1(t)x_2(t)\bar{k} i(t)(\beta_2-\beta_1)+\frac{m\omega}{U_{max}}x_1(t)x_2(t) \label{eq:behavior}\\
& \cdot\left\{u(c_n)\cdot\pmb{\omega}[\bar{k}\beta_1  i(t)] \right. \left.+u(c_1)-u(c_n)\cdot\pmb{\omega}[\bar{k}\beta_2  i(t)]-u(c_2) \right \}.  \nonumber
\end{align}
Note that $x_1(t)+x_2(t)=1$, we have $\frac{dx_2}{dt}=-\frac{dx_1}{dt}$.

Combining \eqref{eq:disease} and \eqref{eq:behavior}, we get the epidemic-behavior co-evolution model in \eqref{eq:M-model}. The first differential equation represents the dynamics of individuals' health states. The second equation models the changes in the proportions of individuals adopting the risky behavior.

\subsection{Steady State}
At the steady states of the above epidemic-behavior co-evolution model, both $i(t)$ and $\{ x_1(t),x_2(t) \}$ evolve to a stable state where there are no further changes in the proportions of infected individuals and the choices of different behaviors. As $x_{2}=1-x_1$, we omit the explicit mention of $x_{2}$ in the representation of the steady state, and use $(i^*,x_1^*)$ to represent the steady state. 
Similar to the work in \cite{2020Decisions} and the Lyapunov's first method \cite{lyapunov1992general}, we define the \emph{steady-state} of the epidemic-behavior co-evolution model as follows.
\begin{definition}
	The steady state $(i^*,x_1^*)$ satisfies:
	\begin{equation}
	\left.\dfrac{di}{dt}\right|_{i=i^*} = 0,\ \left.\dfrac{dx}{dt}\right|_{x_1=x_1^*} =0,\ Re(\lambda_1)<0, \ Re(\lambda_2)<0,
 \label{def:2behavior}
	\end{equation}
	where $\lambda_1$ and $\lambda_2$ are the eigenvalues of the Jacobian matrix 
    \begin{equation}
	\left.\begin{pmatrix}
	\frac{\partial i'}{\partial i} &\frac{\partial i'}{\partial x_1}\\
	\frac{\partial x_1'}{\partial i} &\frac{\partial x_1'}{\partial x_1}\\
	\end{pmatrix}\right|_{i=i^*,x_1=x_1^*},
    \end{equation}
    and $Re(x)$ means the real part of $x$.
\end{definition}

As (\ref{def:2behavior}) are often difficult to solve, we often use numerical methods to find the solution of (\ref{def:2behavior}) and the steady states of the co-evolution of disease spread and behavioral choice.

\section{The steady state Analysis}

During the COVID-19 crisis, governments have been promoting home isolation as a crucial measure to contain the spread of the disease. Home isolation helps significantly reduces the risk of being infected and controls the pandemic, while it also leads to substantial economic loss as well as impacts people's physical and mental well-being. During the pandemic, individuals need to choose between high-cost low-risk conservative behavior and low-cost high-risk risky behavior. Their decisions are often influenced by the severity of the epidemic and the potential loss due to home isolation. People tend to choose self-isolation when the pandemic poses a greater threat to their health; while they may be inclined to choose to go out when the loss due to home isolation is too high (e.g., losing their jobs and income).

Therefore, in this section, based on our model in the previous section, we analyze the dynamics of people's choices between two behaviors: risky behavior (going out) represented by $a_1$, and conservative behavior (home isolation) represented by $a_2$, and the evolution of the epidemic.
We assume that the infection rate for conservative behavior is $\beta_2 = 0$ (as isolation ensures no infection risk). 

Then we analyze the steady state of our model. To find the steady state, 
we first introduce Lemma \ref{lemma1}.
\begin{lemma}
	The necessary and sufficient conditions that $(i^*,x_1^*)$ is a steady state is:
	\begin{equation}
	\begin{split}
	&\left.\dfrac{di}{dt}\right|_{i=i^*} = 0,\left.\dfrac{dx_1}{dt}\right|_{x_1=x_1^*} =0,\\
	&\left.(\dfrac{\partial i'}{\partial i}+\dfrac{\partial x_1'}{\partial x_1})\right|_{i=i^*,x_1=x_1^*}<0,\\ \text{and} \quad
	&\left.(\dfrac{\partial i'}{\partial i}\dfrac{\partial x_1'}{\partial x_1}-\dfrac{\partial i'}{\partial x_1}\dfrac{\partial x_1'}{\partial i})\right|_{i=i^*,x_1=x_1^*}>0.
	\end{split}
	\end{equation}
	\label{lemma1}
\end{lemma}

\begin{proof} By the definition of the steady state, we first have $\frac{di}{dt} = 0$ and $\frac{dx_1}{dt} =0$.  For simplicity, we let 
\begin{equation}
    \begin{aligned}
        &P(i^*,x_1^*) = \left.(\dfrac{\partial i'}{\partial i}+\dfrac{\partial x_1'}{\partial x_1})\right|_{i=i^*,x_1=x_1^*},\\ \mbox{and} \quad
        &Q(i^*,x_1^*) = \left.(\dfrac{\partial i'}{\partial i}\dfrac{\partial x_1'}{\partial x_1}-\dfrac{\partial i'}{\partial x_1}\dfrac{\partial x_1'}{\partial i})\right|_{i=i^*,x_1=x_1^*}.
    \end{aligned}
    \label{eq:PQ}
\end{equation}

Then we have:
\begin{equation}
\lambda_1=\frac{P+\sqrt{P^2-4Q}}{2},\ \mbox{and} \; \lambda_2=\frac{P-\sqrt{P^2-4Q}}{2}.
\end{equation}

From \eqref{def:2behavior}, the steady state requires $Re(\lambda_1)<0$ and $Re(\lambda_2)<0$. If $Q\le0$, then it implies $Re(\lambda_1)\ge0$, which contradicts the steady state conditions. Therefore, we must have $Q>0$. Additionally, since $Q>0$ and $Re(\lambda_1)<0$, it follows that $P<0$. Therefore, we have $P<0$ and $Q>0$, which means $\frac{\partial i'}{\partial i}+\frac{\partial x_1'}{\partial x_1}<0$ and $\frac{\partial i'}{\partial i}\frac{\partial x_1'}{\partial x_1}-\frac{\partial i'}{\partial x_1}\frac{\partial x_1'}{\partial i}>0$. 
\end{proof}

From Lemma \ref{lemma1}, we can have Theorem \ref{th2} about the steady state $(i^*,x_1^*)$.

\begin{theorem}
	The steady state $(i^*,x_1^*)$ of \eqref{eq:M-model} satisfies:
	\begin{equation}
	(i^*,x_1^*)=
	\left\{
	\begin{aligned}
	&(0,1),&& \text{if}\ \bar{k}<\frac{\gamma}{\beta_1},\\
	&\left(1-\frac{\gamma}{\bar{k}\beta_1},1\right),&& \text{if}\ \bar{k}>\frac{\gamma}{\beta_1},\Phi < 0, \\
	&\left(i^{(2)},\frac{\gamma}{(1-i^{(2)})\bar{k}\beta_1}\right),&& \text{if}\ \bar{k}>\frac{\gamma}{\beta_1}, \Phi \ge 0,
	\end{aligned}
	\right.
	\label{eq:steadyPT}
	\end{equation}
	\label{th2}
\end{theorem}
where 

$\Phi {\buildrel \triangle \over =}\ -k_0(u(c_1)-u(c_2))-(\gamma-\bar{k}\beta_1) -k_0u(c_n)\cdot\pmb{\omega} [\bar{k}\beta_1 -\gamma]$, $k_0=\frac{m\omega}{U_{max}}$,

and $i^{(2)}$ is the solution of the equation 
\begin{equation} \label{eqn:i2def}
k_0u(c_n)\cdot\pmb{\omega}[\bar{k}\beta_1i^{(2)}]-\bar{k}\beta_1i^{(2)}+k_0(u(c_1)-u(c_2))=0.
\end{equation}

\begin{proof} By setting $\frac{di}{dt} = 0,\frac{dx_1}{dt} =0$, we can get four steady state candidates: $(0,0)$, $(0,1)$, $\left(1-\frac{\gamma}{\bar{k}\beta_1},1\right)$ and $\left(i^{(2)},\frac{\gamma}{(1-i^{(2)})\bar{k}\beta_1}\right)$. Then, we check whether they are stable. Note that in reality, the payoff of going out should be larger than quarantine, we have $c_1-c_2>0$ (that means $u(c_1)-u(c_2)>0$). Since $c_n$ is the loss of being infected, so $c_n<0$ (that means $u(c_n)<0$). Based on Lemma \ref{lemma1}, the steady state should satisfy $P(i^*,x_1^*)<0$ and $Q(i^*,x_1^*) >0$, where $P(\cdot,\cdot)$ and $Q(\cdot,\cdot)$ are defined in \eqref{eq:PQ}.
Then we have:
\begin{itemize}
	\item For the state $(0,0)$, we have $P=k_0(u(c_1)-u(c_2))-\gamma$ and $Q=-\gamma k_0(u(c_1)-u(c_2))$. Since $u(c_1)-u(c_2)>0$, $\gamma>0$ and $k_0>0$, we have $Q<0$.
	So $(0,0)$ is unstable.
	
	\item For the state $(0,1)$, we have $P=-(\gamma-\bar{k}\beta_1)-k_0(u(c_1)-u(c_2))$ and $Q=k_0(\gamma-\bar{k}\beta_1)(u(c_1)-u(c_2))$. As $u(c_1)-u(c_2)>0$ and $k_0>0$, if $Q>0$, we have $\gamma-\bar{k}\beta_1>0$, which is equivalent to $\bar{k}<\frac{\gamma}{\beta_1}$.  
	When $\gamma-\bar{k}\beta_1>0$, we also have $P<0$.
	Therefore, when $\bar{k}<\frac{\gamma}{\beta_1}$, $(0,1)$ is a steady state.
	
	\item For the state $\left(1-\frac{\gamma}{\bar{k}\beta_1},1\right)$, we have:
	\begin{equation}
	\begin{split}
	P=&-k_0(u(c_1)-u(c_2))-k_0u(c_n)\cdot\pmb{\omega} [\bar{k}\beta_1 -\gamma],\\
	Q=&(\gamma-\bar{k}\beta_1)\{-k_0(u(c_1)-u(c_2))\\
	&-(\gamma-\bar{k}\beta_1)-k_0u(c_n)\cdot\pmb{\omega} [\bar{k}\beta_1 -\gamma]\}.
	\end{split}
	\end{equation}
	
	Since $1-\frac{\gamma}{\bar{k}\beta_1}$ is the proportion of infected individuals, we have $1-\frac{\gamma}{\bar{k}\beta_1}>0$, which is equivalent to $\bar{k}>\frac{\gamma}{\beta_1}$. (If $1-\frac{\gamma}{\bar{k}\beta_1}=0$, then the state $\left(1-\frac{\gamma}{\bar{k}\beta_1},1\right)$ becomes $(0,1)$, which we have discussed).
	To satisfy $P<0$ and $Q>0$, we first define
	\begin{equation}
	\begin{aligned}
	h_1{\buildrel \triangle \over =}&-k_0(u(c_1)-u(c_2))-k_0u(c_n)\cdot\pmb{\omega} [\bar{k}\beta_1 -\gamma],\\ \mbox{and} \;
	h_2{\buildrel \triangle \over =}&-k_0(u(c_1)-u(c_2))-(\gamma-\bar{k}\beta_1)\\
	&-k_0u(c_n)\cdot\pmb{\omega} [\bar{k}\beta_1 -\gamma].
	\end{aligned}
	\end{equation}
	Note that $P<0$ and $Q>0$ are equivalent to $h_1<0$ and $h_2<0$.	
	It is obvious that $h_1=h_2-(\bar{k}\beta_1-\gamma)$.
	Since $(\bar{k}\beta_1-\gamma)>0$, if $h_2<0$, then $h_1<0$, and $P<0$, $Q>0$. Note that $\Phi=h_2$. Therefore, $\left(1-\frac{\gamma}{\bar{k}\beta_1},1\right)$ is a steady state when $\bar{k}>\frac{\gamma}{\beta_1}$ and $\Phi<0$.
	
	\item 

		First, as $i^{(2)}$ and $\frac{\gamma}{(1-i^{(2)})\bar{k}\beta_1}$ are the proportion of infected individuals and the proportion of individuals who choose risky behavior, by definition, we have $0\le i^{(2)} \le 1$ and $0\le \frac{\gamma}{(1-i^{(2)})\bar{k}\beta_1} \le 1$, which is equivalent to:
	\begin{equation}
		0\le i^{(2)} \le 1-\frac{\gamma}{\bar{k}\beta_1}.
		\label{eq:condi}
	\end{equation}
	
	From (\ref{eq:condi}), it is obvious that we should have $0<1-\frac{\gamma}{\bar{k}\beta_1}$, which means $\bar{k}>\frac{\gamma}{\beta_1}$ (If $0=1-\frac{\gamma}{\bar{k}\beta_1}$, then the state $\left(i^{(2)},\frac{\gamma}{(1-i^{(2)})\bar{k}\beta_1}\right)$ becomes $(0,1)$, which we have discussed). Then, we define $f_2(x){\buildrel \triangle \over =} k_0u(c_n)\cdot\pmb{\omega}(\bar{k}\beta_1 x)-\bar{k}\beta_1 x+k_0(u(c_1)-u(c_2))$.
	From \eqref{eqn:i2def}, $i^{(2)}$ is the solution of $f_2(x)$ so we have $f_2(i^{(2)})=0$. It is obvious that $f_2(x)$ is a monotonically decreasing function of $i^{(2)}$. Then $0\le i^{(2)} \le 1-\frac{\gamma}{\bar{k}\beta_1}$ is equivalent to:
	\begin{equation}
		f_2(0)\ge 0,\ \ f_2(1-\frac{\gamma}{\bar{k}\beta_1})\le 0.
	\end{equation}
	
	 Note that $u(c_1)-u(c_2) >0$ and $k_0>0$. Therefore we have $f_2(0)=k_0(u(c_1)-u(c_2))> 0$. So when $\bar{k}>\frac{\gamma}{\beta_1}$ and $f_2(1-\frac{\gamma}{\bar{k}\beta_1})\le 0$, we have $0\le i^{(2)} \le 1$ and $0\le \frac{\gamma}{(1-i^{(2)})\bar{k}\beta_1} \le 1$.
	
	Then, if $\left(i^{(2)},\frac{\gamma}{(1-i^{(2)})\bar{k}\beta_1}\right)$ is a steady state, it should satisfies $P<0$ and $Q>0$, so we have:
	
	\begin{equation}
	\begin{split}
	P=&-\gamma+(1-2i^{(2)})\bar{k}\beta_1x_1^{(2)}=-\bar{k}\beta_1 i^{(2)} x_1^{(2)},\\
	Q=& -\bar{k}\beta_1 i^{(2)} x_1^{(2)}(1-i^{(2)})(1-x_1^{(2)})\\
	&\cdot(k_0u(c_n)\pmb{\omega}'[\bar{k}\beta_1 i^{(2)}]\bar{k}\beta_1-\bar{k}\beta_1).
	\end{split}
	\end{equation}
	
	It is obvious that $P<0$.
	Since $\pmb{\omega}'[\bar{k}\beta_1 i^{(2)}] > 0$ and $k_0u(c_n)<0$, we have $k_0u(c_n)\pmb{\omega}'[\bar{k}\beta_1 i^{(2)}]\bar{k}\beta_1-\bar{k}\beta_1<0$. So $Q>0$. 
	
	Therefore, when $\bar{k}>\frac{\gamma}{\beta_1}$ and $f_2(1-\frac{\gamma}{\bar{k}\beta_1})\le 0$, $\left(i^{(2)},\frac{\gamma}{(1-i^{(2)})\bar{k}\beta_1}\right)$ is steady state. Note that $f_2(1-\frac{\gamma}{\bar{k}\beta_1})=-\Phi$, therefore, when 
	$\bar{k}>\frac{\gamma}{\beta_1}$ and $\Phi\ge 0$, $\left(i^{(2)},\frac{\gamma}{(1-i^{(2)})\bar{k}\beta_1}\right)$ is steady state.
\end{itemize}
\end{proof}

\begin{figure*}[t]
	\subfigure[The proportion of infected individuals]{
		\begin{minipage}[t]{0.5\linewidth}
			\centering
			\includegraphics[width=0.8\linewidth]{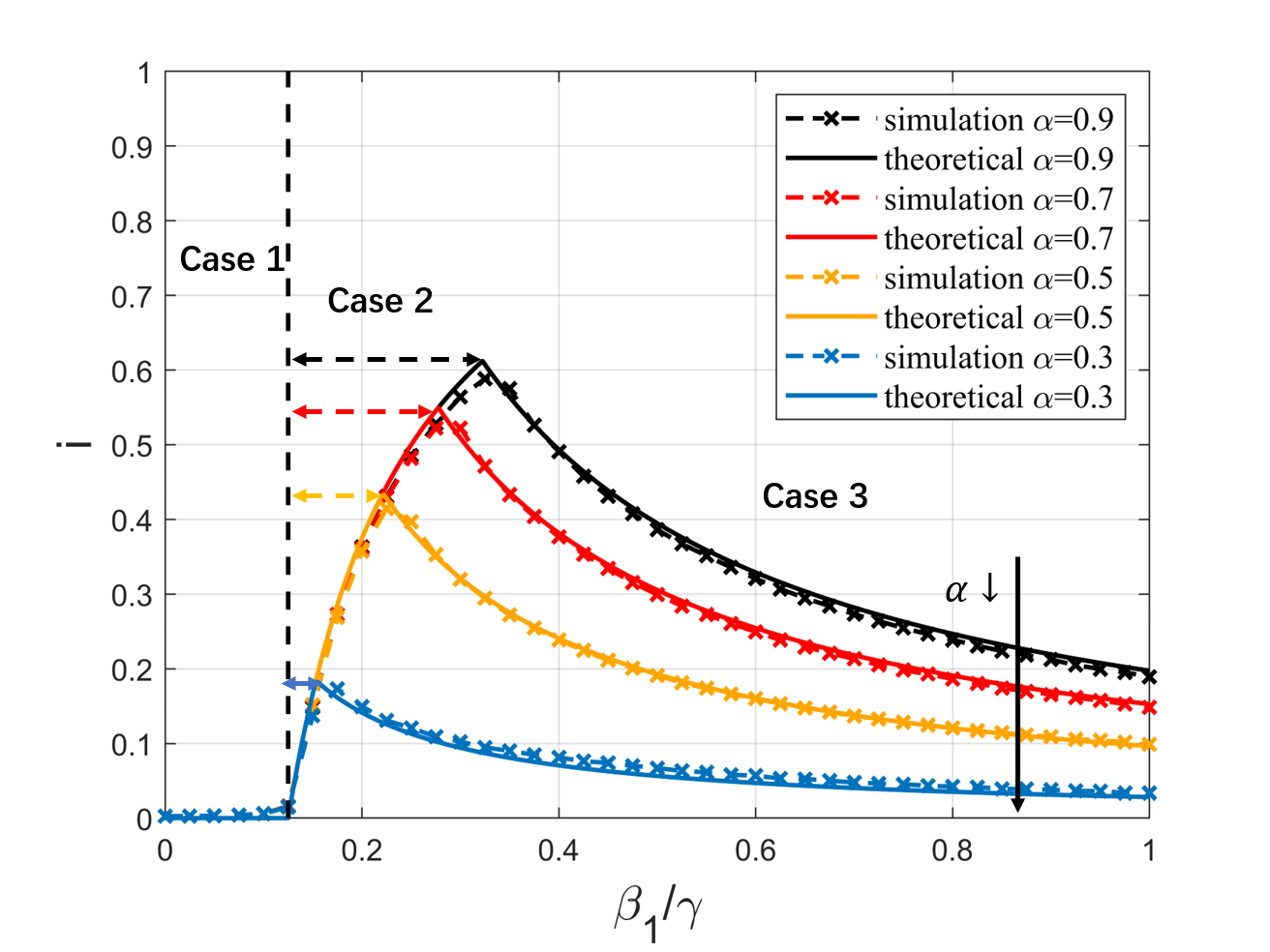}
			\label{change_beta1}
		\end{minipage}%
	}%
	\subfigure[The proportion of individuals choosing risky behavior]{
		\begin{minipage}[t]{0.5\linewidth}
			\centering
			\includegraphics[width=0.8\linewidth]{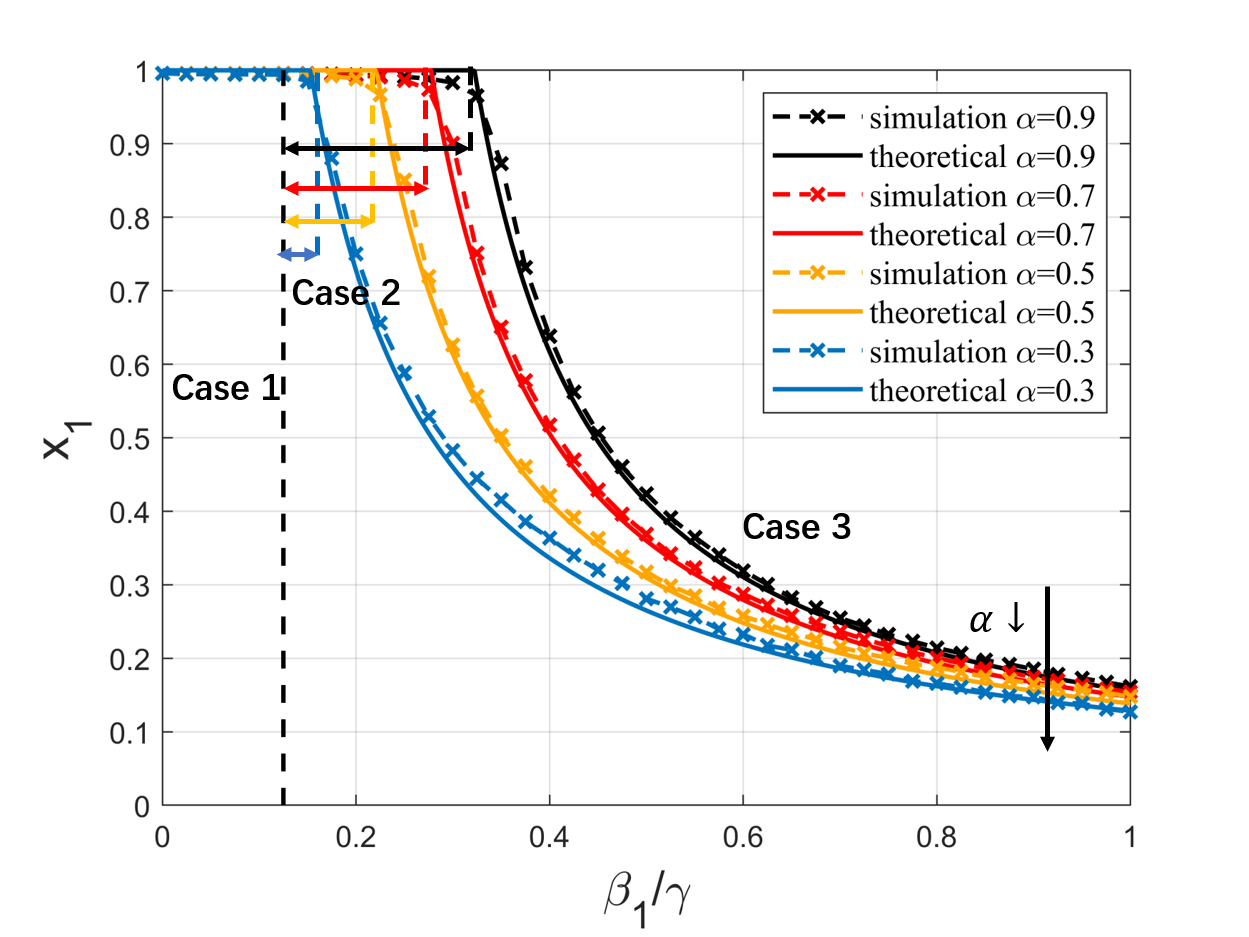}
			\label{change_beta2}
		\end{minipage}%
	}%
	
	\centering
	\caption{The simulation and theoretical results of steady states with different disease spread ability
	}
	\label{change_beta}
\end{figure*}

\begin{figure*}[t]
	\centering
	\subfigure[The proportion of infected individuals]{
		\begin{minipage}[t]{0.5\linewidth}
			\centering
			\includegraphics[width=0.8\linewidth]{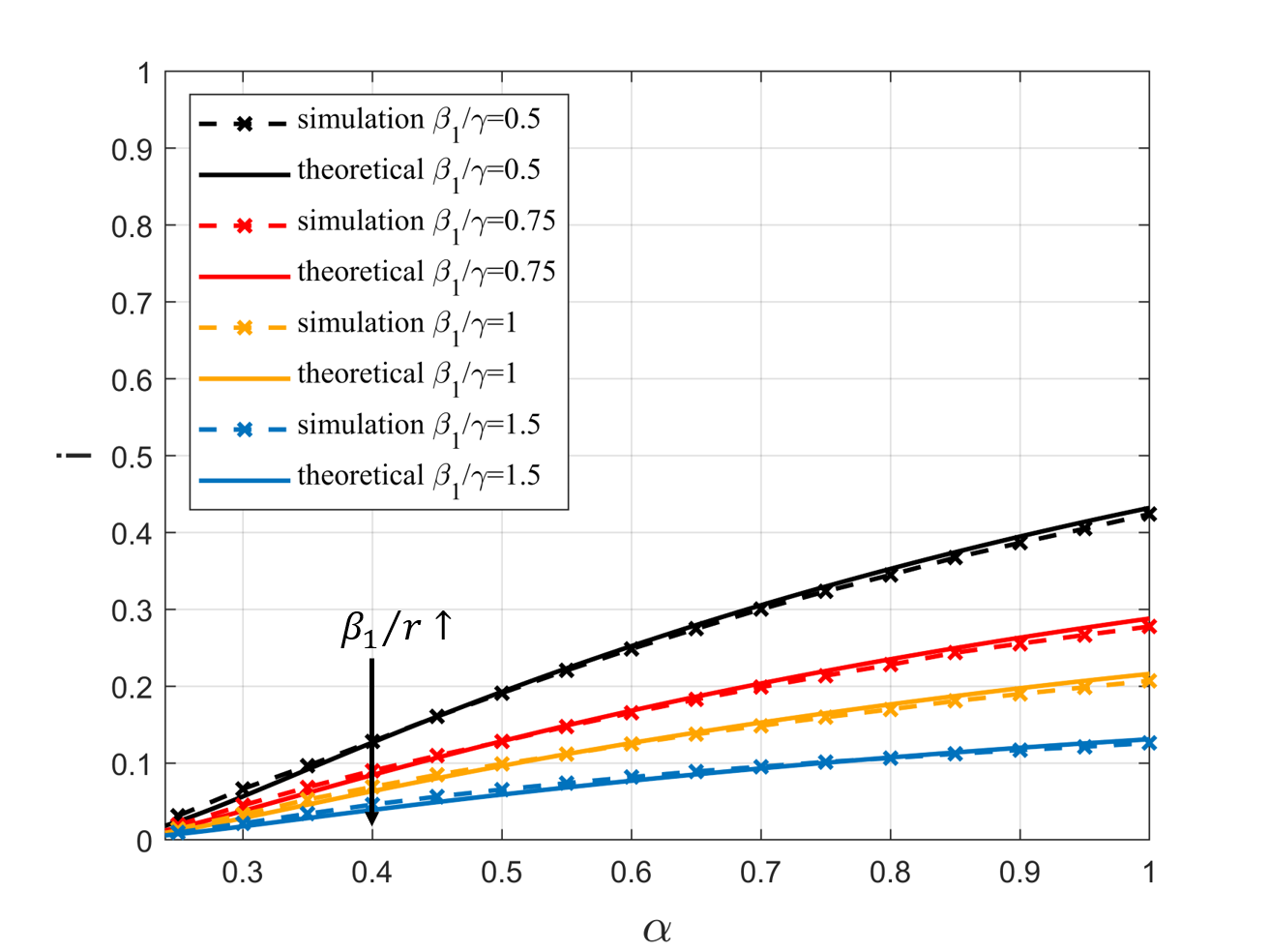}
			\label{change_alpha_low1}
		\end{minipage}%
	}%
	\subfigure[The proportion of individuals choosing risky behavior]{
		\begin{minipage}[t]{0.5\linewidth}
			\centering
			\includegraphics[width=0.8\linewidth]{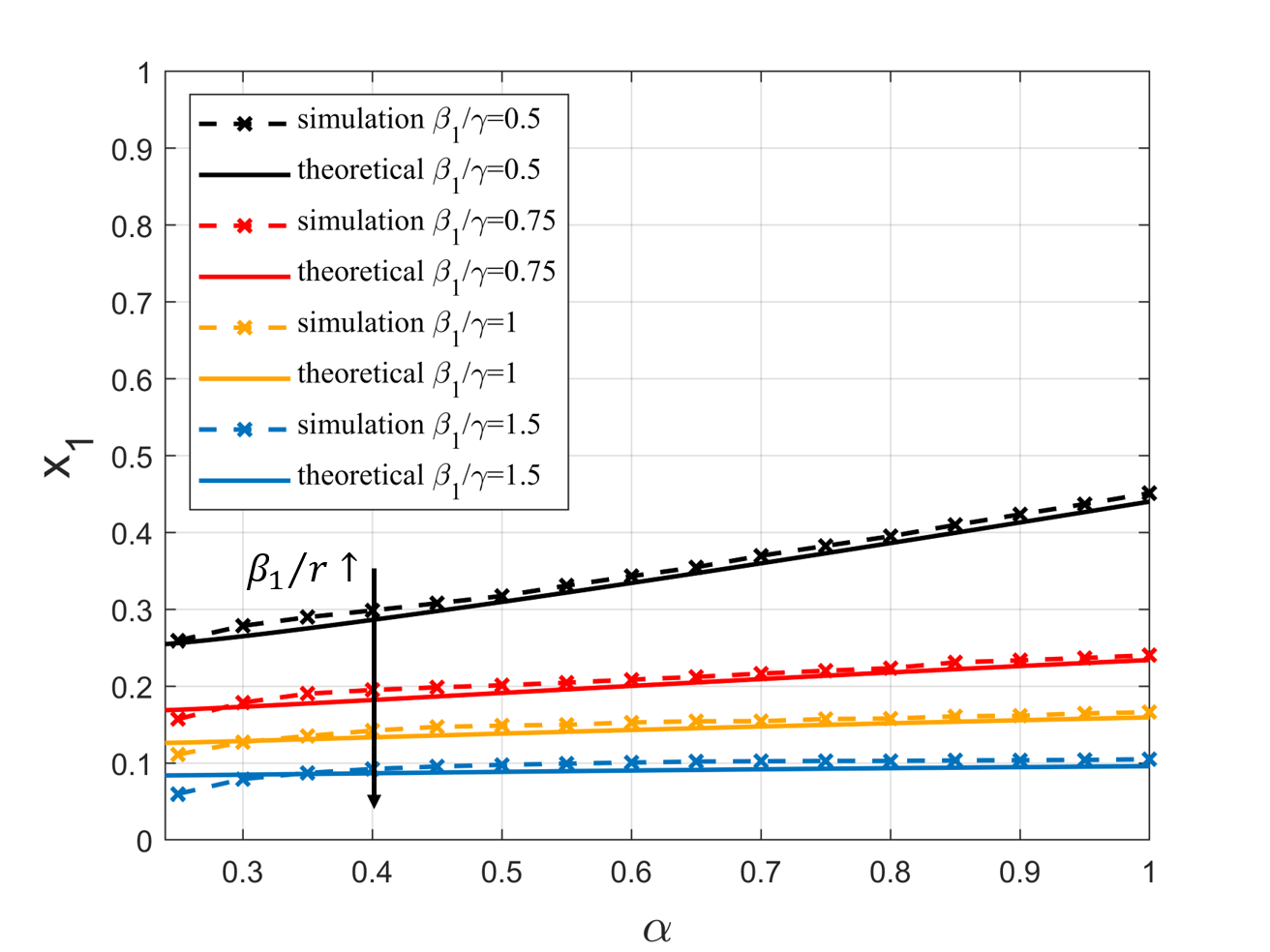}
			\label{change_alpha_low2}
		\end{minipage}%
	}%

        \centering
	\caption{The simulation and theoretical results of steady states with different disease rationality coefficients $\alpha$.
	}
 \label{alpha}
 \end{figure*}
From Theorem \ref{th2}, our epidemic-behavior co-evolution model has possible three steady states, and each corresponds to a particular situation in reality:
\begin{itemize}
	\item Case 1: 
The steady state $(0,1)$ represents an extreme situation where the infection rate is too low, and the disease would eventually die out even without any intervention or change of behavior. In this case, all individuals choose risky behavior. For the steady state $(0,1)$, the stability condition is $\bar{k}<\frac{\gamma}{\beta_1}$, which is equivalent to $\frac{\beta_1}{\gamma}<\frac{1}{\bar{k}}$. Here, $\frac{1}{\bar{k}}$ represents the epidemic threshold of a regular network with degree $\bar{k}$ \cite{2007Epidemic}. If $\frac{\overline{\beta}}{\gamma}<\frac{1}{\bar{k}}$, the epidemic would die out; otherwise, it would spread out \cite{2007Epidemic}. Since $\overline{\beta} \in [0,\beta_1]$, the disease would die out regardless of how people choose their behavior. In this situation, all individuals would choose the risky behavior in the steady state. Therefore, when $\bar{k}<\frac{\gamma}{\beta_1}$, the stable point is $i=0$ and $x_1=1$.
 
 

	\item Case 2: 
For the steady state $\left(1-\frac{\gamma}{\bar{k}\beta_1},1\right)$, the first constraint is $\bar{k}>\frac{\gamma}{\beta_1}$, which means if all individuals choose the risky behavior, the disease will spread out. The second constraint is $\Phi < 0$. When all individuals choose to go out with $x_1=1$, the percentage of infected individuals will reach the stable state $\bar{i}=1-\frac{\gamma}{\bar{k}\beta_1}$, which is the maximum extent to which the disease can spread (please refer to Appendix \ref{adx:A} for a proof). If for all possible $i$ in the range $[0, \bar{i}]$, we have $\frac{dx_1}{dt}|_{x_1\neq 0,1}>0$ (if $x_1=0 \text{ or } 1$, $\frac{dx_1}{dt}=0$), then more people will choose risky behavior as time goes on, and $x_1=1$ with all individuals choosing the risky behavior at the steady state. This may happen when the risky behavior's payoff is much higher than the conservative behavior's with $c_1\gg c_2$, or when the cost of being infected $c_n$ is very low. If $\frac{dx_1}{dt}|_{i=\bar{i},x_1\neq 0,1}>0$, which is equivalent to $\Phi<0$, then for all $i\in[0, \bar{i}]$, we have $\frac{dx_1}{dt}|_{x_1\neq 0,1}>0$. Therefore, when $\bar{k}>\frac{\gamma}{\beta_1}$ and $\Phi<0$, $\left(1-\frac{\gamma}{\bar{k}\beta_1},1\right)$ is the steady state.

	
	\item Case 3: 
 The steady state $\left(i^{(2)},\frac{\gamma}{(1-i^{(2)})\bar{k}\beta_1}\right)$ represents the general situation other than the two extreme cases. In this scenario, the disease does not extinct, nor does it spread to the maximum extent, and at the steady state, $0<i^{(2)}<\bar{i}$ percent of individuals will be infected. Meanwhile, $0<\frac{\gamma}{(1-i^{(2)})\bar{k}\beta_1}<1$ of susceptible individuals will choose the risky behavior. When $\bar{k}>\frac{\gamma}{\beta_1}$, and $\Phi \geq 0$, the state $\left(i^{(2)},\frac{\gamma}{(1-i^{(2)})\bar{k}\beta_1}\right)$ is stable.
	
	
		
\end{itemize}

\section{Simulation Results}
In this section, we run simulations to verify our analysis in the above sections. The network we used is a regular network with 500 nodes and a fixed degree of 8. Since the risky behavior of going out is the default behavior that most people take in their daily lives, we set the payoff of this behavior to 0, while the conservative behavior of self-isolation can be regarded as behavior with losses and a negative payoff. As an example, we set the payoff of the risky behavior as $c_1=0$, the payoff for the conservative behavior as $c_2=-1$, and the disease loss as $c_n=-10$. For the value function, we set $\sigma=0.65$ according to \cite{prelec2000compound}. 
In our simulations, we set $\lambda=1$, and observe the same trend for other values of $\lambda$.

Fig. \ref{change_beta} shows the results of steady state with changing $\frac{\beta_1}{\gamma}$ which reflects the spread ability of the disease. A higher value of $\frac{\beta_1}{\gamma}$ indicates a greater ability of the disease to spread within the population. First, from Fig. \ref{change_beta}, the theoretical results fit well with the simulation results. The three cases in Theorem \ref{th2} can be clearly seen in Fig. \ref{change_beta}. When $\frac{\beta_1}{\gamma}$ is low, even if all individuals choose risky behavior, the disease will die out, which corresponds to Case 1 in the figure. When $\frac{\beta_1}{\gamma}$ is high enough so that the disease would spread out, it enters Case 2. As $\frac{\beta_1}{\gamma}$ increases, the proportion of infected individuals in the steady state gradually increases, and at this time, all individuals would still choose risky behaviors since it has a higher payoff. As $\frac{\beta_1}{\gamma}$ continues increasing, it enters Case 3, where more and more people choose self-isolation due to the high risk of infection. In this case, although $\frac{\beta_1}{\gamma}$ become higher, more people choose to self-isolate due to the increased risk of infection, which makes the proportion of infected individuals in the steady state gradually decrease.

Moreover, from Fig. \ref{change_beta}, given $\frac{\beta_1}{\gamma}$, individuals with a lower rationality coefficient $\alpha$ (meaning that such individuals are more irrational) have a smaller probability to be in Case 2, 
which means that irrational individuals are more inclined to choose conservative behaviors during a pandemic. This may be because from the Prospect Theory, irrational users tend to overestimate the small probability of being infected, and thus they tend to be more conservative.
We also run simulations on the steady state of Case 3 under different rationality coefficients $\alpha$ and the results are shown in Fig. \ref{alpha}. It can be seen that fewer individuals choose risky behavior when the rational coefficient $\alpha$ is smaller, confirming our previous conclusion about their conservative behavior. 

\section{Conclusion}
In this paper, we propose the epidemic-behavior co-evolution method to capture individuals' behavior during a pandemic and analyze the co-evolution of disease spread and behavior change. We incorporate irrational behavior based on Prospect Theory and examine how it influences people's decision-making processes and, in turn, the spread of the disease. Through rigorous theoretical analysis, we identify three steady states under different scenarios and investigate their practical implications. To validate our theoretical findings and gain further insights, we conduct simulation experiments. The experimental results align closely with our theoretical predictions, confirming the correctness of our model. Notably, our findings consistently indicate that irrationality tends to steer individuals towards more conservative behavioral choices during a pandemic, effectively mitigating the spread of the disease.

\appendices

\section{Proof for the maximum spread extent}
\label{adx:A}
For the steady state $(i^*,x_1^*)$, it satisfies $\frac{di}{dt}=0$. If $i^*\neq 0$, then based on \eqref{eq:M-model}, we have:
\begin{equation}
    i^*=1-\frac{\gamma}{\bar{k}\beta_1x_1^*},
\end{equation}
and $i^*$ is an increasing function of $x_1^*$, where $x_1^* \in [0,1]$. Therefore, $i^*$ reaches its maximum value when $x_1^*=1$. Substituting this value into the equation, we obtain $\bar{i}=1-\frac{\gamma}{\bar{k}\beta_1}$.

\thispagestyle{noheader}
\bibliographystyle{IEEEtran}
\bibliography{mylib}

\begin{thebibliography}{10}
\providecommand{\url}[1]{#1}
\csname url@samestyle\endcsname
\providecommand{\newblock}{\relax}
\providecommand{\bibinfo}[2]{#2}
\providecommand{\BIBentrySTDinterwordspacing}{\spaceskip=0pt\relax}
\providecommand{\BIBentryALTinterwordstretchfactor}{4}
\providecommand{\BIBentryALTinterwordspacing}{\spaceskip=\fontdimen2\font plus
\BIBentryALTinterwordstretchfactor\fontdimen3\font minus
  \fontdimen4\font\relax}
\providecommand{\BIBforeignlanguage}[2]{{%
\expandafter\ifx\csname l@#1\endcsname\relax
\typeout{** WARNING: IEEEtran.bst: No hyphenation pattern has been}%
\typeout{** loaded for the language `#1'. Using the pattern for}%
\typeout{** the default language instead.}%
\else
\language=\csname l@#1\endcsname
\fi
#2}}
\providecommand{\BIBdecl}{\relax}
\BIBdecl

\bibitem{zhang2014suppression}
H.-F. Zhang, J.-R. Xie, M.~Tang, and Y.-C. Lai, ``Suppression of epidemic
  spreading in complex networks by local information based behavioral
  responses,'' \emph{Chaos: An Interdisciplinary Journal of Nonlinear Science},
  vol.~24, no.~4, p. 043106, 2014.

\bibitem{wu2012impact}
Q.~Wu, X.~Fu, M.~Small, and X.-J. Xu, ``The impact of awareness on epidemic
  spreading in networks,'' \emph{Chaos: an interdisciplinary journal of
  nonlinear science}, vol.~22, no.~1, p. 013101, 2012.

\bibitem{bagnoli2007risk}
F.~Bagnoli, P.~Lio, and L.~Sguanci, ``Risk perception in epidemic modeling,''
  \emph{Physical Review E}, vol.~76, no.~6, p. 061904, 2007.

\bibitem{funka2009spread}
S.~Funka, E.~Gilada, C.~Watkinsb, and V.~A. Jansena, ``The spread of awareness
  and its impact on epidemic outbreaks,'' \emph{PNAS}, vol. 106, no.~16, 2009.

\bibitem{granell2013dynamical}
C.~Granell, S.~G{\'o}mez, and A.~Arenas, ``Dynamical interplay between
  awareness and epidemic spreading in multiplex networks,'' \emph{Physical
  review letters}, vol. 111, no.~12, p. 128701, 2013.

\bibitem{granell2014competing}
------, ``Competing spreading processes on multiplex networks: awareness and
  epidemics,'' \emph{Physical review E}, vol.~90, no.~1, p. 012808, 2014.

\bibitem{wang2019impact}
Z.~Wang, Q.~Guo, S.~Sun, and C.~Xia, ``The impact of awareness diffusion on
  sir-like epidemics in multiplex networks,'' \emph{Applied Mathematics and
  Computation}, vol. 349, pp. 134--147, 2019.

\bibitem{zheng2018interplay}
C.~Zheng, C.~Xia, Q.~Guo, and M.~Dehmer, ``Interplay between sir-based disease
  spreading and awareness diffusion on multiplex networks,'' \emph{Journal of
  Parallel and Distributed Computing}, vol. 115, pp. 20--28, 2018.

\bibitem{kahneman2013prospect}
D.~Kahneman and A.~Tversky, ``Prospect theory: An analysis of decision under
  risk,'' in \emph{Handbook of the fundamentals of financial decision making:
  Part I}.\hskip 1em plus 0.5em minus 0.4em\relax World Scientific, 2013, pp.
  99--127.

\bibitem{kahneman2013choices}
------, ``Choices, values, and frames,'' in \emph{Handbook of the fundamentals
  of financial decision making: Part I}.\hskip 1em plus 0.5em minus 0.4em\relax
  World Scientific, 2013, pp. 269--278.

\bibitem{tversky1992advances}
A.~Tversky and D.~Kahneman, ``Advances in prospect theory: Cumulative
  representation of uncertainty,'' \emph{Journal of Risk and uncertainty},
  vol.~5, no.~4, pp. 297--323, 1992.

\bibitem{oraby2015bounded}
T.~Oraby and C.~T. Bauch, ``Bounded rationality alters the dynamics of
  paediatric immunization acceptance,'' \emph{Scientific reports}, vol.~5,
  no.~1, pp. 1--12, 2015.

\bibitem{hota2019game}
A.~R. Hota and S.~Sundaram, ``Game-theoretic vaccination against networked sis
  epidemics and impacts of human decision-making,'' \emph{IEEE Transactions on
  Control of Network Systems}, vol.~6, no.~4, pp. 1461--1472, 2019.

\bibitem{li2020perception}
X.-J. Li and X.~Li, ``Perception effect in evolutionary vaccination game under
  prospect-theoretic approach,'' \emph{IEEE Transactions on Computational
  Social Systems}, vol.~7, no.~2, pp. 329--338, 2020.

\bibitem{poletti2009spontaneous}
P.~Poletti, B.~Caprile, M.~Ajelli, A.~Pugliese, and S.~Merler, ``Spontaneous
  behavioural changes in response to epidemics,'' \emph{Journal of theoretical
  biology}, vol. 260, no.~1, pp. 31--40, 2009.

\bibitem{prelec2000compound}
D.~Prelec, ``Compound invariant weighting functions in prospect theory,''
  \emph{Choices, values, and frames}, pp. 67--92, 2000.

\bibitem{1979Prospect}
K.~A. Tversky, ``Prospect theory: An analysis of decision under risk,''
  \emph{Econometrica}, vol.~47, no.~2, pp. 263--291, 1979.

\bibitem{prelec1998probability}
D.~Prelec, ``The probability weighting function,'' \emph{Econometrica}, pp.
  497--527, 1998.

\bibitem{0DIRECTED}
J.~O. Kephart and S.~R. White, \emph{DIRECTED-GRAPH EPIDEMIOLOGICAL MODELS OF
  COMPUTER VIRUSES}.\hskip 1em plus 0.5em minus 0.4em\relax Computation: The
  Micro And The Macro View.

\bibitem{wu2020nowcasting}
J.~T. Wu, K.~Leung, and G.~M. Leung, ``Nowcasting and forecasting the potential
  domestic and international spread of the 2019-ncov outbreak originating in
  wuhan, china: a modelling study,'' \emph{The Lancet}, vol. 395, no. 10225,
  pp. 689--697, 2020.

\bibitem{2004Coevolutionary}
A.~Traulsen, J.~C. Claussen, and C.~Hauert, ``Coevolutionary dynamics: From
  finite to infinite populations,'' 2004.

\bibitem{2020Decisions}
J.~Rowlett and C.~J. Karlsson, ``Decisions and disease: the evolution of
  cooperation in a pandemic,'' 2020.

\bibitem{lyapunov1992general}
A.~M. Lyapunov, ``The general problem of the stability of motion,''
  \emph{International journal of control}, vol.~55, no.~3, pp. 531--534, 1992.

\bibitem{2007Epidemic}
D.~Chakrabarti, Y.~Wang, C.~Wang, J.~Leskovec, and C.~Faloutsos, ``Epidemic
  thresholds in real networks,'' \emph{Acm Transaction on Information \& System
  Security}, vol.~10, no.~4, pp. p.1--26, 2007.

\end{thebibliography}
\end{document}